\documentclass[letterpaper, 10 pt, conference]{ieeeconf}  

\IEEEoverridecommandlockouts                              

\usepackage[noadjust]{cite}

\usepackage{amsmath}
\usepackage{amsthm}
\usepackage{amsfonts}
\usepackage[ruled,vlined,linesnumbered]{algorithm2e}
\usepackage{algpseudocode}
\usepackage{xfrac}
\usepackage{listings}
\usepackage{bbm}
\usepackage{color}
\usepackage{soul}
\usepackage{booktabs}
\usepackage{array}
\usepackage{dsfont}
\usepackage[multiple]{footmisc}

\newcommand{\agentIdxA}{i}
\newcommand{\agentIdxB}{j}
\newcommand{\iterIdx}{q}
\newcommand{\iterIdxB}{p}
\newcommand{\numAgent}{M}
\newcommand{\agentSet}{\mathcal{M}}
\newcommand{\col}{\text{col}}

\DeclareMathOperator*{\argmin}{arg\,min}

\newtheorem{theorem}{Theorem}

\newtheorem{proposition}{Proposition}

\newtheorem{corollary}{Corollary}

\theoremstyle{definition}
\newtheorem{assumption}{Assumption}

\theoremstyle{definition}

\theoremstyle{remark}
\newtheorem{remark}{Remark}

\title{\LARGE \bf
Trajectory Optimization for Nonlinear Multi-Agent Systems using Decentralized Learning Model Predictive Control*
}

\author{Edward L.\ Zhu$^{1}$, Yvonne R.\ St\"urz$^{1}$, Ugo Rosolia$^{2}$, and Francesco Borrelli$^{1}$
\thanks{$^{1}$The authors are with the Department of Mechanical Engineering,
        University of California at Berkeley, Berkeley, CA 94701 USA
        {\tt\small edward.zhu@berkeley.edu}}
\thanks{$^{2}$The author is with the Department of Mechanical and Civil Engineering,
California Institute of Technology, Pasadena, CA 91125 USA}
\thanks{*This project has received funding from the European Union’s Horizon 2020 research and innovation programme under the Marie Sklodowska-Curie grant agreement No.\ 846421. This work was partially funded by the grant ONR-N00014-18-1-2833.}
}

\begin{document}

\maketitle
\thispagestyle{empty}
\pagestyle{empty}

\begin{abstract}
We present a decentralized minimum-time trajectory optimization scheme based on learning model predictive control for multi-agent systems with nonlinear decoupled dynamics and coupled state constraints. By performing the same task iteratively, data from previous task executions is used to construct and improve local time-varying safe sets and an approximate value function. These are used in a decoupled MPC problem as terminal sets and terminal cost functions. Our framework results in a decentralized controller, which requires no communication between agents over each iteration of task execution, and guarantees persistent feasibility, finite-time closed-loop convergence, and non-decreasing performance of the global system over task iterations. Numerical experiments of a multi-vehicle collision avoidance scenario demonstrate the effectiveness of the proposed scheme.

\end{abstract}

\section{Introduction} \label{sec:intro}

In this paper, we study the problem of decentralized Model Predictive Control (MPC) for dynamically decoupled multi-agent systems under the minimum-time cost and coupled state constraints. Multi-agent systems typically exhibit inter-agent coupling, which can be expressed as constraints on the global system. MPC is a well-studied approach to the control of such constrained systems and can be applied in a global manner for centralized control of multi-agent systems with a small number of agents. However, as the number of agents increases, centralized approaches typically become intractable in practice due to limitations in computational power and communication capacities~\cite{negenborn2009multi}. This gives rise to decentralized and distributed MPC schemes, which leverage the inherent parallelizable structure of multi-agent systems to reduce the required computational effort. Feasibility and stability of MPC are typically obtained using a terminal cost function and terminal constraints in the MPC design~\cite{mayne2000constrained}, which we refer to as terminal components. However, synthesis of these terminal components for the control of nonlinear multi-agent systems is in general challenging.

The primary advantage of decentralized MPC over its distributed counterpart is the lower communication demand. While only local information is used in the control of each agent in decentralized schemes, additional communication between agents is required to obtain a control action in distributed methods \cite{bemporad2010networked,negenborn2014distributed,dunbar2006distributed,ferranti2018coordination,conte2016distributed,luis2019trajectory,liu2010distributed, Stuerz2020}. A literature review on distributed MPC is extensive and goes beyond the scope of this conference paper. 

In decentralized MPC, local controllers are synthesized for each agent, where feasibility and stability properties are attained via robustness of the controller against coupling to other agents. In prior work, conditions for the stability of decentralized MPC for nonlinear systems are investigated. In \cite{alessio2007decentralized,alessio2011decentralized}, this is attained via terminal cost synthesis and an online supervisory scheme for modifying the decoupling structure to meet these conditions without destabilizing the system. The work in~\cite{keviczky2006decentralized} achieves this by bounding the prediction error of neighboring agents' states. However, this method assumes the singleton terminal set of the origin, which limits the controller's domain of attraction. In addtion, neither method can deal with coupling state constraints between the agents. 
 
Iterative approaches to the control of multi-agent systems have also been proposed. Of such methods, Sequential Convex Programming (SCP) \cite{boyd2008sequential} is similar to the approach proposed in this paper. SCP solves a non-convex optimization problem by successively forming convex approximations about previous solutions. In \cite{augugliaro2012generation}, SCP was used to generate collision-free trajectories for multiple quadcopters in a centralized manner. This was extended into a decentralized formulation in \cite{chen2015decoupled}, which showed improvements in computational tractability. However, in these works, SCP is a \textit{heuristic} method and may fail to find a feasible solution. In addition, SCP optimizes over the entire trajectory, which contrasts with the receding horizon approach taken in this work, and may be computationally challenging for long time horizons or fine time discretizations.

In this paper, we propose a decentralized approach to trajectory optimization for nonlinear multi-agent systems. By performing the same task over iterations, we collect data on the agents' closed-loop behavior to successively improve the construction of terminal components for the local controllers. In particular, as we improve an estimate of the terminal cost (an approximate value function) and expand the terminal constraint set (the domain of this function), we iteratively improve closed-loop performance of the multi-agent system, while maintaining feasibility and finite-time convergence guarantees. 

The contribution of this work is twofold. We extend the method of Learning Model Predictive Control (LMPC) from \cite{rosolia2017learning} to the multi-agent case for dynamically decoupled agents under the minimum-time cost and coupled state constraints. In particular, we first propose a procedure for synthesizing the MPC terminal components using data from previous iterations of task execution. We then show that the resulting decentralized LMPC has the properties of persistent feasibility, finite-time closed-loop convergence to the goal state, and non-decreasing performance over iterations. We demonstrate the effectiveness of the decentralized method with a numerical example in the context of multi-vehicle collision avoidance, where we observe a significant reduction of computational effort compared to a centralized approach.
 
\section{Preliminaries} \label{sec:preliminaries}

\subsection{System Description} \label{sec:system_description}
For a system of $\numAgent$ agents, the set of indices $\{1, \dots, \numAgent\}$ is denoted as $\agentSet$. `$\preceq$' denotes an element-wise inequality. Consider the global nonlinear time-invariant discrete-time system composed of $\numAgent$ agents \vspace{-0.1cm}
\begin{align}
	x_{t+1} = f(x_{t}, u_{t}), \vspace{-0.1cm}
	\label{eq:nonlin_dyn_global}
\end{align} 
where the global state and input vectors at sampling time $t \geq 0$ are formed by stacking those from each agent into a single column, i.e. $x_t = \col_{\agentIdxA \in \agentSet}(x_{\agentIdxA,t}) = [x_{1,t}^{\top}, \dots, x_{M,t}^{\top}]^{\top} \in \mathbb{R}^n$ and $u_t = \col_{\agentIdxA \in \agentSet}(u_{\agentIdxA,t}) \in \mathbb{R}^m$, where $x_{\agentIdxA,t} \in \mathbb{R}^{n_i}$ and $u_{\agentIdxA,t} \in \mathbb{R}^{m_i}$. Each agent in the global system is subject to local state and input constraints, \vspace{-0.1cm}
\begin{equation}
    x_{\agentIdxA,t} \in \mathcal{X}_\agentIdxA \subseteq \mathbb{R}^{n_\agentIdxA}, \quad u_{\agentIdxA,t} \in \mathcal{U}_\agentIdxA \subseteq \mathbb{R}^{m_\agentIdxA}, \quad \forall t \geq 0.
    \label{eq:local_constr}
\end{equation}
These local constraint sets are assumed to be closed, compact, and include the goal states $x_{\agentIdxA,F}$ in their relative interiors. The global system is additionally subject to coupling constraints on the system state, \vspace{-0.1cm}
\begin{equation}
    g(x_t) \preceq 0, \quad \forall t \geq 0.
    \label{eq:coupled_constr}
\end{equation}

We assume that the agents are dynamically decoupled with continuous dynamics and locally stabilizable, which means that we can write (interchangeably with \eqref{eq:nonlin_dyn_global}), for all agents $\agentIdxA \in \agentSet$, the local dynamics as
\begin{align}
	x_{\agentIdxA,t+1} = f_\agentIdxA(x_{\agentIdxA,t}, u_{\agentIdxA,t}).
	\label{eq:nonlin_dyn_local}
\end{align}

\subsection{Control Objective}

The objective is to design a controller which drives the system state to a goal state $x_F$ by solving the following optimal control problem for the global system 
\begin{subequations} \label{eq:global_ocp}
    \begin{align}
       \min_{\mathbf{u}_T} \quad & \sum_{t=0}^{T}h(x_t, u_t) \label{eq:global_ihocp_cost}\\
        \text{s.t.} \quad & \ x_{t+1} = f(x_t, u_t), \quad \forall t \in \{0, \dots, T-1\} \\
        & \ x_0 = x_S, \ x_T = x_F \\
        & \ x_t \in \mathcal{X}, \ u_t \in \mathcal{U} \quad \forall t \in \{0, \dots, T\} \label{eq:global_ihocp_local_constrs}\\
        & \ g(x_t) \preceq 0, \quad \forall t \in \{0, \dots, T\}, \label{eq:global_ihocp_couple_constrs}
    \end{align}
\end{subequations}
where the goal state $x_F = \col_{\agentIdxA \in \agentSet}(x_{\agentIdxA,F})$ is assumed to be a feasible equilibrium state of \eqref{eq:nonlin_dyn_global}. $x_S= \col_{\agentIdxA \in \agentSet}(x_{\agentIdxA,S})$ is the initial condition of the system, and the state and input constraint sets in \eqref{eq:global_ihocp_local_constrs} are the Cartesian products of the local constraint sets in \eqref{eq:local_constr}. In this work, we are specifically interested in the minimum-time cost, which is defined as follows
\begin{align}
    h(x_t,u_t) = \mathds{1}(x;x_F) = \begin{cases}
    0 & \quad \text{if} \ x = x_F \\
    1 & \quad \text{otherwise}.
    \end{cases} \nonumber
\end{align}

\subsection{Learning Model Predictive Control} \label{sec:lmpc}

In this section, we briefly introduce and review the key concepts of LMPC \cite{rosolia2017learning,rosolia2019minimum}, which are extended upon in this work. LMPC was proposed as an iterative trajectory optimization method for single-agent nonlinear dynamical systems performing iterative tasks. In particular, the method provides a data-driven approach to terminal set and cost synthesis. We assume that an initial feasible input sequence $\mathbf{u}^0 = \{u_0^0, u_1^0, \dots, u_{T^0-1}^0\}$ and closed-loop state trajectory $\mathbf{x}^0 = \{x_0^0, x_1^0, \dots, x_{T^0}^0\}$ exists for \eqref{eq:nonlin_dyn_global} and is available at iteration $\iterIdx = 0$. $T^\iterIdx$ denotes the time at which the closed-loop system reaches the terminal state, i.e. $x_{T^\iterIdx}^\iterIdx = x_F$. We note that for iterative tasks, the initial condition of the system is assumed to be the same over iterations, i.e. $x_0^\iterIdx = x_S, \ \forall \iterIdx \geq 0$. 

LMPC solves a Finite Horizon Optimal Control Problem (FHOCP), which approximates \eqref{eq:global_ocp}, in a receding horizon fashion. Given the initial condition $x$ at time $t$ of iteration $\iterIdx$, the FHOCP is defined as
\begin{subequations}
    \begin{align}
        \min_{\mathbf{u}_{t,N}^{\iterIdx}} \ & \sum_{k=0}^{N-1}\mathds{1}(x_{k|t}^\iterIdx; x_F) + V^{\iterIdx-1}(x_{N|t}^\iterIdx) \label{eq:lmpc_cost} \\
        \text{s.t.} \ & \ x_{k+1|t}^\iterIdx = f(x_{k|t}^\iterIdx, u_{k|t}^\iterIdx), \ \forall k \in \mathbb{N}_{N-1} \\
        & \ x_{0|t}^\iterIdx = x \\
        & \ x_{k|t}^\iterIdx \in \mathcal{X}, \ u_{k|t}^\iterIdx \in \mathcal{U}, \ \forall k \in \mathbb{N}_{N-1} \label{eq:lmpc_local_constrs}\\
        & \ g(x_{k|t}^\iterIdx) \preceq 0, \ \forall k \in \mathbb{N}_{N-1} \label{eq:lmpc_coupled_constrs} \\
        & \ x_{N|t}^\iterIdx \in \mathcal{SS}^{\iterIdx-1}, \label{eq:lmpc_term_constr}
    \end{align}
\end{subequations}
where $x_{k|t}^{\iterIdx}$ and $u_{k|t}^{\iterIdx}$ denote the decision variables of the predicted state and input at the sampling time $t+k$ of iteration $\iterIdx$. The first input $u_{0|t}^{\iterIdx}$ is then applied to the system~\eqref{eq:nonlin_dyn_global}.

The terminal set at iteration $\iterIdx-1$ in \eqref{eq:lmpc_term_constr}, called a sampled safe set, is defined as
\begin{equation}
    \mathcal{SS}^{\iterIdx-1} = \left\{\bigcup_{\iterIdxB \in \mathcal{Q}^{\iterIdx-1}}\bigcup_{t = 0}^{T^\iterIdxB} x_t^\iterIdxB\right\},
    \label{eq:single_agent_ss}
\end{equation}
where $\mathcal{Q}^{\iterIdx-1} = \{\iterIdxB \in \{0,\dots,\iterIdx-1\} : x_{T^\iterIdxB}^p = x_F\}$ is the set of iteration indices up to iteration $\iterIdx-1$ where the goal state $x_F$ was successfully reached. The sampled safe set collects the closed-loop state trajectories from previous successful iterations, which implies that at iteration $\iterIdx$, $\forall x \in \mathcal{SS}^{\iterIdx}$, there exists a known and feasible sequence of inputs, with $x_0 = x$, such that $f(x_t,u_t) \in \mathcal{SS}^{\iterIdx}, \ \forall t \geq 0$. We note that by construction, $\mathcal{SS}^{\iterIdx}$ is a control invariant set.

For the terminal cost function in \eqref{eq:lmpc_cost}, an approximate value function is constructed which returns the minimum cost-to-go (over iterations 0 to $\iterIdx-1$) from each state in the safe set. The cost-to-go from the state at time $t$ along the
closed-loop trajectory $\mathbf{x}^\iterIdx$ with corresponding input sequence $\mathbf{u}^\iterIdx$ is defined as $J_{T^{\iterIdx}}^{\iterIdx}(x_t^{\iterIdx},t) = \sum_{k=t}^{T^{\iterIdx}}\mathds{1}(x_k^{\iterIdx}; x_F)$.

The approximate value function at iteration $\iterIdx-1$ is then
\begin{align}
    V^{\iterIdx-1}(x) = \begin{cases}
        \min_{(p,t) \in \mathcal{F}^{\iterIdx-1}(x)} J_{T^{\iterIdx}}^{\iterIdxB}(x,t) \ &\text{if} \ x \in \mathcal{SS}^{\iterIdx-1} \\
        + \infty \ &\text{otherwise}
    \end{cases}
    \label{eq:single_agent_value_func}
\end{align}
where $\mathcal{F}^{\iterIdx-1}(x) = \{(\iterIdxB,t) : x_t^\iterIdxB = x \ \text{and} \ x_t^\iterIdxB \in \mathcal{SS}^{\iterIdx-1}\}$ returns the set of iteration and time index pairs for the states in previous trajectories which are equal to $x$.

Using the properties of the constructed terminal components, the resulting scheme guarantees persistent feasibility of the FHOCP and finite-time closed-loop convergence of the closed-loop system. Non-decreasing performance  can also be shown over iterations of task execution.

\section{Problem Formulation} \label{sec:prob_form}

In this section, we extend the idea of LMPC from Sec.~\ref{sec:lmpc} to the multi-agent case. Specifically, our formulation leverages data from previous iterations of task execution to synthesize local controllers for each agent via decoupled FHOCPs. This allows for an entirely decentralized control scheme for the multi-agent system at each iteration. In the following, we describe the synthesis of the local terminal components, namely the time-varying sampled safe sets and approximate value function. We then present the resulting decoupled FHOCPs, which are solved in a receding horizon manner. Combining these elements, we arrive at a decentralized LMPC procedure for trajectory optimization of multi-agent systems.

\subsection{Time-Varying Sampled Safe Sets and Global Constraint Decomposition} \label{sec:local_tv_fi_set}

In this part, we address two challenges posed by a decentralized receding horizon approach. Namely 1) how global constraint satisfaction can be enforced through decentralized local constraint satisfaction, and 2) how feasibility over the entire task horizon can be maintained in a receding horizon implementation of \eqref{eq:global_ocp}, particularly when time-varying constraints are introduced. 

\emph{To address 1)}, we introduce the following assumption. Notice that the original global time-invariant constraints are transformed into local time-varying constraints.

\begin{assumption}
    At iteration $\iterIdx$, there exists a time-varying local decomposition of the global constraints: $g_{\agentIdxA,t}^{\iterIdx}(\cdot), \ \forall t \geq 0, \ \forall \agentIdxA \in \agentSet$, which can be constructed from feasible trajectories of the global system, At each time $t$, joint local constraint satisfaction is sufficient for global constraint satisfaction, i.e. $g_{\agentIdxA,t}^{\iterIdx}(x_{\agentIdxA,t}) \preceq 0, \ \forall \agentIdxA \in \agentSet \implies g(x_t) \preceq 0$.
    \label{asm:constraint_decomposition}
\end{assumption}

\begin{remark} \label{rem:asm1_rem}
    Given a feasible trajectory $\mathbf{x}_{\agentIdxA}^{\iterIdx}$, we may always obtain the decomposition $g_{\agentIdxA,t}^{\iterIdx}(x_{\agentIdxA}) = \|x_{\agentIdxA}-x_{\agentIdxA,t}^{\iterIdx}\|$ for each trajectory point in $\mathbf{x}_{\agentIdxA}^{\iterIdx}$, which satisfies the condition in Ass.~\ref{asm:constraint_decomposition}. In Sec.~\ref{sec:collision_avoidance_example}, we present a technique for constructing $g_{\agentIdxA,t}^{\iterIdx}(\cdot)$ in a less conservative manner.
\end{remark}

\emph{To address 2)}, constraint satisfaction of the FHOCPs over the entire task horizon is obtained by using control invariant sets. This was done in \cite{rosolia2017learning} via \eqref{eq:single_agent_ss}, which, recall, are control invariant sets by definition. We would like to achieve the same goal of using data to construct sets which grant the same properties to the decentralized scheme. However, in \eqref{eq:global_ocp} (and due to the constraint decomposition in Assumption~\ref{asm:constraint_decomposition}), we include time-varying constraints for which the classical definition of set invariance does not apply. For this reason, we do not collect all recorded states of the global system, which converge to $x_{F}$, into a single invariant set as in \eqref{eq:single_agent_ss}. Instead, we interpret the states of agent $\agentIdxA$ at time $t$ of previous successful iterations as a sampled subset of the $(\hat{T}^{\iterIdx-1}-t)$-step reachable set to $x_{\agentIdxA,F}$, where $\hat{T}^\iterIdx=\max_{\iterIdxB\in\{0,\dots,\iterIdx\}}(T_\agentIdxA^\iterIdxB)$. As such, for each agent $\agentIdxA\in\agentSet$ at iteration $\iterIdx$, we propose to construct time-varying sampled safe sets $\mathcal{SS}_{\agentIdxA,t}^{\iterIdx}$ using data from previous iterations of task execution, which are one-step reachable to $\mathcal{SS}_{\agentIdxA,t+1}^{\iterIdx}$ (and therefore $K$-step reachable to $x_{\agentIdxA,F}$ for some $K \in \{1,\dots, \hat{T}^{\iterIdx-1}-t\}$). Moreover, we require that the global coupling constraints are satisfied for all combinations of states in the constructed safe sets, i.e. $\forall t \geq 0$ and $\bar{x}_{\agentIdxA} \in \mathcal{SS}_{\agentIdxA,t}^{\iterIdx}$, $g(\bar{x})  \preceq 0$, where $\bar{x}=\col_{\agentIdxA \in \agentSet}(\bar{x}_{\agentIdxA})$. The following assumption allows us to start with non-empty safe sets at the first iteration.

\begin{assumption}
    At iteration $\iterIdx = 0$, an initial feasible input and state sequence, which converges to the goal state in $T_\agentIdxA^0$ steps, exists for each agent (\ref{eq:nonlin_dyn_local}). Denote this input sequence and state trajectory for agent $\agentIdxA \in \agentSet$ as $\mathbf{u}_{\agentIdxA}^{0}$ and $\mathbf{x}_{\agentIdxA}^{0}$. 
    \label{asm:init_feas_traj}
\end{assumption}

\begin{remark} \label{rem:asm1_rem}
    This assumption is reasonable as one may provide initial feasible trajectories through demonstration or by employing a conservative controller.
\end{remark}

Let $(\mathcal{D}_{x},\mathcal{D}_{u})$ be the dataset which collects the state trajectories $\mathbf{x}_{\agentIdxA}^{0} \dots, \mathbf{x}_{\agentIdxA}^{\iterIdx}$ and input sequences $\mathbf{u}_{\agentIdxA}^{0}, \dots, \mathbf{u}_{\agentIdxA}^{\iterIdx}$, and let $g_{\agentIdxA,t}^{\iterIdx}(\cdot)$ be a constraint decomposition which satisfies Assumption~\ref{asm:constraint_decomposition}. Then, initializing the recursive relationship with $\mathcal{SS}_{\agentIdxA,\hat{T}^{\iterIdx-1}}^{\iterIdx} = \{x_{\agentIdxA,F}\}$, we construct the time-varying sampled safe sets for each agent $\agentIdxA$ at time $t$ of iteration $\iterIdx$ as
\begin{align}
    \mathcal{SS}_{\agentIdxA,t}^{\iterIdx} = \left\{ x_{\agentIdxA,k}^{\iterIdxB} \in \mathbf{x} : \ g_{\agentIdxA,t}^{\iterIdx}(x_{\agentIdxA,k}^{\iterIdxB}) \preceq 0 \ \text{and} \right. \nonumber \\ 
    \left. f_{\agentIdxA}(x_{\agentIdxA,k}^{\iterIdxB},u_{\agentIdxA,k}^{\iterIdxB}) \in \mathcal{SS}_{\agentIdxA,t+1}^{\iterIdx} \right\}.
    \label{eq:sampled_safe_set_synth}
\end{align}

The time-varying sampled safe sets $\mathcal{SS}_{\agentIdxA,t}^{\iterIdx}$ collect the state trajectory points from previous successful iterations for which there exists a trajectory to the goal state that satisfies the global constraints. We show in Sec.~\ref{sec:decent_lmpc_props} that the sampled safe sets are a non-empty family of reachable sets to $x_{\agentIdxA,F}$. 

\begin{remark}
For successful iterations, i.e. where $x_{\agentIdxA,T_\agentIdxA^\iterIdx}^\iterIdx = x_{\agentIdxA,F}$, it is straightforward to accommodate task iterations of different lengths when constructing $\mathcal{SS}_{\agentIdxA,t}^{\iterIdx}$ for $t \in \{0, \dots, \max_{\iterIdxB\in\mathcal{I}}(T_\agentIdxA^\iterIdxB)\}$. Since $x_{\agentIdxA,F}$ is an equilibrium of system $\agentIdxA$, we may trivially apply the zero input to obtain $\mathcal{SS}_{\agentIdxA,t}^{\iterIdx} = \{x_{\agentIdxA,F}\}$ for all $t > T_\agentIdxA^{\iterIdx}$.
\end{remark}

\begin{remark} \label{rem:asm2_rem}
    How the constraint decomposition is constructed depends on the specific problem at hand. In the multi-vehicle collision avoidance example shown in Sec.~\ref{sec:collision_avoidance_example}, we propose a procedure to construct hyperplanes which separate the position states of the sampled safe sets for pairs of agents. 
\end{remark}

\subsection{Value Function Approximation} \label{sec:local_tv_val_func_approx}

For an input sequence $\mathbf{u}_{\agentIdxA}^{\iterIdx}$ and closed-loop state trajectory $\mathbf{x}_{\agentIdxA}^{\iterIdx}$ of agent $\agentIdxA$ at iteration $\iterIdx$ with length $T_\agentIdxA^\iterIdx$, we define the cost-to-go from the state $x_{\agentIdxA,t}^{\iterIdx}$ at time $t$ along the closed-loop trajectory to be
\begin{equation}
    J_{\agentIdxA,T_\agentIdxA^\iterIdx}^{\iterIdx}(x_{\agentIdxA,t}^{\iterIdx}, t) = \sum_{k=t}^{T_\agentIdxA^\iterIdx} \mathds{1}(x_{\agentIdxA,k}^{\iterIdx}, x_{\agentIdxA,F}) = T_\agentIdxA^\iterIdx - t.
    \label{eq:cost_to_go}
\end{equation} 
The iteration cost for agent $\agentIdxA$ at iteration $\iterIdx$ can then be written as $J_{\agentIdxA,T_\agentIdxA^\iterIdx}^{\iterIdx}(x_{\agentIdxA,S},0)$, recalling that the system is initialized at the same state $x_S$ for each iteration $q$. This leads to the definition of the approximate value function $V_{\agentIdxA}^{\iterIdx}(\cdot,t)$ at time $t$ over the sampled safe set $\mathcal{SS}_{\agentIdxA,t}^{\iterIdx}$ as 
\begin{align}
    V_{\agentIdxA}^{\iterIdx}(x_{\agentIdxA},t) = \begin{cases}
        \min_{(\iterIdxB,k) \in \mathcal{F}_{\agentIdxA,t}^{\iterIdx}(x_{\agentIdxA})} J_{\agentIdxA,T_{\agentIdxA}^{\iterIdxB}}^{\iterIdxB}(x_{\agentIdxA},k), & \text{if} \ x_{\agentIdxA} \in \mathcal{SS}_{\agentIdxA,t}^{\iterIdx} \\
        +\infty, & \text{otherwise,}
    \end{cases}
    \label{eq:approx_val_func}
\end{align}
where $\mathcal{F}_{\agentIdxA,t}^{\iterIdx}(x_{\agentIdxA})$ is defined in the same way as in (\ref{eq:single_agent_value_func}).
Thus, for a state $x_{\agentIdxA}$ whose value is equal to some state in $\mathcal{SS}_{\agentIdxA,t}^{\iterIdx}$, $V_{\agentIdxA}^{\iterIdx}(x_{\agentIdxA},t)$ returns the minimum cost-to-go over trajectories which pass through $x_{\agentIdxA}$ in $\mathcal{SS}_{\agentIdxA,k}^{\iterIdx}$ for $k \geq t$. Note that $V_\agentIdxA^\iterIdx$ is discontinuous as its domain is a finite set of points.

\subsection{The Finite Horizon Optimal Control Problem} \label{sec:fhocp}


Synthesis of the terminal components is achieved using the function \texttt{synthesizeFHOCP} as summarized in Alg.~\ref{alg:fhocp_synthesis}, which acts as a global coordinator between iterations of task execution in our decentralized framework. In this algorithm, we introduce the design parameters $\bar{t},\underline{t} \geq 0$ and $\underline{q} \in \{0,\dots,q\}$ to reduce the computational burden of safe set construction and to make satisfaction of one-step reachability from $\mathcal{SS}_{\agentIdxA,t}^{\iterIdx}$ to $\mathcal{SS}_{\agentIdxA,t+1}^{\iterIdx}$ trivial. In particular, we construct the set $\mathcal{I}$ which contains the iteration indices of the $\underline{\iterIdx}$ most recent successful iterations and the sliding time range $\mathcal{T}_t = \{\max(t-\underline{t},0), \dots, t, \dots, t+\bar{t}\}$ and obtain the \textit{candidate} safe sets for time $t$ (line 6). After a constraint decomposition is constructed (line 8, see Rem.~\ref{rem:asm2_rem} and Sec.~\ref{sec:collision_avoidance_example}), the candidate safe sets are checked for constraint satisfaction (line 10). In the case where any constraint is violated, the design parameters are updated (lines 11-15) and new candidate safe sets are constructed (line 16). This procedure is repeated until safe sets satisfying \eqref{eq:sampled_safe_set_synth} are found. We will show in Sec.~\ref{sec:decent_lmpc_props} that the procedure described in Alg. \ref{alg:fhocp_synthesis} results in non-empty safe sets which satisfy the coupled constraints and are reachable to $x_{\agentIdxA,F}$. As such, after performing the task at iteration $\iterIdx-1$, we obtain for agent $\agentIdxA$ the decomposition of the global constraint $g_{\agentIdxA,t}^{\iterIdx-1}(\cdot)$, the sampled safe sets $\mathcal{SS}_{\agentIdxA,t}^{\iterIdx-1}$, and the approximate value function $V_{\agentIdxA}^{\iterIdx-1}(\cdot,t)$, which are constructed using data from iterations 0 to $\iterIdx-1$.

To obtain the control action for agent $\agentIdxA$ at sampling time $t$ of iteration $\iterIdx$, we solve the following decoupled FHOCP with horizon length $N$ and initial condition $x_{\agentIdxA}$.
\begin{subequations}\label{eq:local_fhocp}
    \begin{align}
         \mathcal{P}_{\agentIdxA,N}^{\iterIdx}(x_{\agentIdxA},t)=\min_{\mathbf{u}_{\agentIdxA,t}^{\iterIdx}} ~ &\sum_{k=0}^{N-1} \! \mathds{1}(x_{\agentIdxA,k|t}^{\iterIdx};x_{\agentIdxA,F}) \!+\! V_\agentIdxA^{\iterIdx-1}(x_{\agentIdxA,N|t}^{\iterIdx},t\!+\!N) \nonumber\\
        \text{s.t.} ~~ & \ x_{\agentIdxA,k+1|t}^{\iterIdx} = f_\agentIdxA(x_{\agentIdxA,k|t}^{\iterIdx}, u_{\agentIdxA,k|t}^{\iterIdx}),  \label{eq:local_fhocp_dyn} \\
        & \ x_{\agentIdxA,0|t}^{\iterIdx} = x_{\agentIdxA} \label{eq:local_fhocp_init} \\
        & \ x_{\agentIdxA,k|t}^{\iterIdx} \in \mathcal{X}_\agentIdxA, \ u_{\agentIdxA,k|t}^{\iterIdx} \in \mathcal{U}_\agentIdxA,  \label{eq:local_fhocp_local_constr} \\
        & \ g_{\agentIdxA,t+k}^{\iterIdx-1}(x_{\agentIdxA,k|t}^{\iterIdx}) \preceq 0,  \label{eq:local_fhocp_decoup_constr} \\
        & \ x_{\agentIdxA,N|t}^{\iterIdx} \in \mathcal{SS}_{\agentIdxA,t+N}^{\iterIdx-1}\label{eq:local_fhocp_term_set}\\
        & \ \forall k \in \{0,\dots,N-1\} \nonumber
    \end{align}
\end{subequations}
where \eqref{eq:local_fhocp_dyn} and \eqref{eq:local_fhocp_init} represent the system dynamics and initial condition, respectively. The local state and input constraints are given in \eqref{eq:local_fhocp_local_constr}. \eqref{eq:local_fhocp_decoup_constr} enforces satisfaction of the decomposed constraint for each agent, which is sufficient for global constraint satisfaction. Finally, \eqref{eq:local_fhocp_term_set} ensures that the terminal state is a member of the time-varying sampled safe set $\mathcal{SS}_{\agentIdxA,t+N}^{\iterIdx-1}$. We denote the locally optimal value of the FHOCP cost in \eqref{eq:local_fhocp} as $J_{\agentIdxA,N}^{*,\iterIdx}(x_{\agentIdxA},t)$.

Let $\mathbf{u}_{\agentIdxA,t}^{*,\iterIdx}(x_{\agentIdxA}) = \{u_{\agentIdxA,0|t}^{*,\iterIdx}(x_{\agentIdxA}), \dots, u_{\agentIdxA,N-1|t}^{*,\iterIdx}(x_{\agentIdxA})\}$ denote the input sequence which minimizes (\ref{eq:local_fhocp}) for initial state $x_{\agentIdxA}$ at sampling time $t$, and $\mathbf{x}_{\agentIdxA,t}^{*,\iterIdx} = \{x_{\agentIdxA,0|t}^{*,\iterIdx}, \dots, x_{\agentIdxA,N|t}^{*,\iterIdx}\}$ be the corresponding state trajectory beginning at $x_{\agentIdxA,0|t}^{*,\iterIdx} = x_{\agentIdxA}$. In typical \textit{receding horizon} fashion, for each agent $\agentIdxA \in \agentSet$, the first element of $\mathbf{u}_{\agentIdxA,t}^{*,\iterIdx}(x_{\agentIdxA})$ is applied to system (\ref{eq:nonlin_dyn_local}), which defines the state feedback policy 
\begin{equation}
    u_{\agentIdxA,t}^{\iterIdx} = \kappa_{\agentIdxA}^{\iterIdx}(x_{\agentIdxA},t) \doteq u_{\agentIdxA,0|t}^{*,\iterIdx}(x_{\agentIdxA}),
    \label{eq:lmpc_feedback_control_law}
\end{equation}
with $x_{\agentIdxA} = x_{\agentIdxA,t}^{\iterIdx}$. This results in the closed-loop state trajectory $\mathbf{x}_{\agentIdxA}^{\iterIdx} = \{x_{\agentIdxA,0}^{\iterIdx}, x_{\agentIdxA,1}^{\iterIdx}, \dots, x_{\agentIdxA,t}^{\iterIdx}, \dots\}$ and input sequence $\mathbf{u}_{\agentIdxA}^{\iterIdx} = \{u_{\agentIdxA,0}^{\iterIdx}, u_{\agentIdxA,1}^{\iterIdx}, \dots, u_{\agentIdxA,t}^{\iterIdx}, \dots\}$ for agent $\agentIdxA$ at iteration $\iterIdx$. We will show that the closed-loop state trajectory under \eqref{eq:lmpc_feedback_control_law} converges to the goal state $x_{\agentIdxA,F}$ in finite-time.

\setlength{\textfloatsep}{3pt}
\setlength{\floatsep}{10pt}
\begin{algorithm}[t!]
    \SetAlgoLined
	\KwIn{$\mathcal{D}_{x}$, $\mathcal{D}_{u}$, $\mathbf{t}$, $\theta= \{\underline{\iterIdx},\underline{t},\bar{t}\}$}
	$\check{\iterIdx} \leftarrow \underline{\iterIdx}$\;
	$\mathcal{I} \leftarrow \{\max(\iterIdx - \underline{\iterIdx}), \dots, \iterIdx\}$\;
	\For{$t \in \{0,\dots,\max_{\iterIdxB\in\mathcal{I}}(T_\agentIdxA^{\iterIdxB})\}$}{
	    \For{$\agentIdxA \in \agentSet$}{
	       $\mathcal{T}_{t} \leftarrow \{\max(t-\underline{t},0), \dots, t, \dots, t+\bar{t}\}$\;
	        $\mathcal{SS}_{\agentIdxA,t}^{\iterIdx} \leftarrow \bigcup_{\iterIdxB \in \mathcal{I}} \left\{ x_{\agentIdxA,k}^{\iterIdxB} \in \mathcal{D}_{x} : k \in \mathcal{T}_{t} \right\}$\;
	    }
	    $\{g_{\agentIdxA,t}^{\iterIdx}(\cdot)\}_{\agentIdxA \in \agentSet} \leftarrow$ obtain decomposition according to Ass.~\ref{asm:constraint_decomposition}\;
	}
	 
	\While{any constraint in $g_{\agentIdxA,t}^{\iterIdx}(x)$ is violated for any $x \in \mathcal{SS}_{\agentIdxA,t}^{\iterIdx}, \ \agentIdxA \in \mathcal{M}, \ t \in \{0,\dots,\max_{\iterIdxB\in\mathcal{I}}(T_\agentIdxA^{\iterIdxB})\}$}{
	    $\underline{\iterIdx} \leftarrow \underline{\iterIdx} - 1$\;
	    \If{$\underline{\iterIdx} = 0$}{
	        $\bar{t} \leftarrow \max(\bar{t}-1,0)$, $\underline{t} \leftarrow \max(\underline{t}-1,0)$\;
	        $\underline{\iterIdx} \leftarrow \check{\iterIdx}$\;
	    }
	    Repeat lines 2-9 with updated $\underline{t}$, $\bar{t}$ and $\underline{\iterIdx}$\;
	}
	$V_{\agentIdxA}^{\iterIdx}(\cdot,t) \leftarrow$ compute as in (\ref{eq:cost_to_go}) and (\ref{eq:approx_val_func}) $\forall x_{\agentIdxA}^{\iterIdx} \in \mathcal{SS}_{\agentIdxA,t}^{\iterIdx}$\;
	\KwOut{$\{\mathcal{SS}_{\agentIdxA,t}^{\iterIdx},V_{\agentIdxA}^{\iterIdx}(\cdot,t),g_{\agentIdxA,t}^{\iterIdx}(\cdot)\}$}
	\caption{\texttt{synthesizeFHOCP}}
	\label{alg:fhocp_synthesis}
\end{algorithm} 

\begin{algorithm}[t!]
	\SetAlgoLined
	\KwIn{$Z$, $\mathbf{x}^{0}$, $\mathbf{u}^{0}$, $\{T_\agentIdxA^0\}_{\agentIdxA\in\agentSet}$, $\theta$}
	$\mathcal{D}_{x} \leftarrow \{\mathbf{x}^{0}\}$, $\mathcal{D}_{u} \leftarrow \{\mathbf{u}^{0}\}$, $\mathbf{t} \leftarrow \{T_\agentIdxA^0\}_{\agentIdxA\in\agentSet}$\;
	$\{\mathcal{SS}_{\agentIdxA,t}^{0}, V_{\agentIdxA}^{0}(\cdot,t), g_{\agentIdxA,t}^{0}(\cdot)\} \leftarrow$ \texttt{synthesizeFHOCP}$(\mathcal{D}_{x}, \mathcal{D}_{u}, \mathbf{t}, \theta)$\;
	\For{$\iterIdx \in \{1, \dots, Z\}$}{
	    \For{$\agentIdxA \in \agentSet$ in parallel}{
	        $t \leftarrow 0$, $x_{\agentIdxA,0}^{\iterIdx} \leftarrow x_{\agentIdxA,S}$, $\mathbf{x}_{\agentIdxA}^{\iterIdx} \leftarrow \{x_{\agentIdxA,S}\}$, $\mathbf{u}_{\agentIdxA}^{\iterIdx} \leftarrow \emptyset$\;
    	   \While{$x_{\agentIdxA,t}^{\iterIdx} \neq x_{\agentIdxA,F}$}{
        	    $\kappa_{\agentIdxA}^{\iterIdx}(\cdot,t) \leftarrow$ solve $\mathcal{P}_{\agentIdxA,N}^{\iterIdx}(x_{\agentIdxA,t}^{\iterIdx},t)$\;
        	    $x_{\agentIdxA,t+1}^{\iterIdx} \leftarrow$ apply $u_{\agentIdxA,t}^{\iterIdx} = \kappa_{\agentIdxA}^{\iterIdx}(x_{\agentIdxA,t}^{\iterIdx},t)$ and measure state\;
        	    $\mathbf{x}_{\agentIdxA}^{\iterIdx} \leftarrow \mathbf{x}_{\agentIdxA}^{\iterIdx} \cup \{x_{\agentIdxA,t+1}^{\iterIdx}\}$, $\mathbf{u}_{\agentIdxA}^{\iterIdx} \leftarrow \mathbf{u}_{\agentIdxA}^{\iterIdx} \cup \{u_{\agentIdxA,t}^{\iterIdx}\}$\;
        	    $t \leftarrow t + 1$\;
    	    }
    	    $\mathcal{D}_{x} \leftarrow \mathcal{D}_{x} \cup \{\mathbf{x}_{\agentIdxA}^{\iterIdx}\}$, $\mathcal{D}_{u} \leftarrow \mathcal{D}_{u} \cup \{\mathbf{u}_{\agentIdxA}^{\iterIdx}\}$\;
    	    $T_{\agentIdxA}^{\iterIdx} \leftarrow t$
    	}
    	$\mathbf{t} \leftarrow \mathbf{t} \cup \{T_{\agentIdxA}^{\iterIdx}\}_{\agentIdxA \in \agentSet}$\;
	    $\{\mathcal{SS}_{\agentIdxA,t}^{\iterIdx}, V_{\agentIdxA}^{\iterIdx}(\cdot,t), g_{\agentIdxA,t}^{\iterIdx}(\cdot)\} \leftarrow$ \texttt{synthesizeFHOCP}$(\mathbf{x}, \mathbf{u}, \mathbf{t}, \theta)$\;
	}
	\caption{Decentralized Learning Model Predictive Control}
	\label{alg:decent_lmpc}
\end{algorithm}

\begin{remark}
Due to the construction of the sampled safe sets as collections of discrete trajectory points, $\mathcal{P}_{\agentIdxA,N}^{\iterIdx}(x_{\agentIdxA},t)$ is a mixed integer nonlinear program, which can be computationally intensive to solve. However, the minimum-time formulation can be exploited for computational efficiency, e.g. by parallelization and branch and bound methods. Certain nonlinear systems may also admit a convex relaxation of the sampled safe set and approximate value function while maintaining performance guarantees \cite{rosolia2019minimum}.
\end{remark}

\subsection{Decentralized Learning Model Predictive Control} \label{sec:decent_lmpc}

The resulting iterative LMPC scheme for the multi-agent system is described in Alg.~\ref{alg:decent_lmpc}, where the for loop beginning at line 3 corresponds to the iterations of the decentralized LMPC, the for loop over agents from lines 4 to 14 may be executed in an entirely decentralized manner with no communication between agents, and the while loop from lines 6 to 11 iterates over time steps. Note that we assume that there is no mismatch between the model used in \eqref{eq:local_fhocp_dyn} and the system on which the policy is applied (line 8). The inputs to Alg.~\ref{alg:decent_lmpc} are $Z \geq 1$: the number of LMPC iterations, $\mathbf{x}^{0}$ and $\mathbf{u}^{0}$: the initial feasible state and input sequences, $\{T_\agentIdxA^0\}_{\agentIdxA\in\agentSet}$: the lengths of the initial trajectories, and $\theta$ which contains the design parameters for \texttt{synthesizeFHOCP}. It is implicitly assumed that each iteration is successful. We will show that this is true in Sec.~\ref{sec:decent_lmpc_props}.

\section{Properties of the Decentralized LMPC} \label{sec:decent_lmpc_props}
In this section, we show that for each iteration $\iterIdx$, the sampled safe sets as constructed in (\ref{eq:sampled_safe_set_synth}) have the property of reachability to $x_{\agentIdxA,F}$, which is sufficient for persistent feasibility of the decentralized LMPC over the entire task horizon for all iterations $\iterIdx$. Furthermore, we show that the closed-loop system converges to the goal state in finite time, and that task performance is non-decreasing over iterations.

\begin{proposition} \label{prop:backwards_reachable}
    Let Assumption~\ref{asm:init_feas_traj} hold, then for agent $\agentIdxA$ at iteration $\iterIdx$, the sampled safe sets as constructed in (\ref{eq:sampled_safe_set_synth}) using Alg.~\ref{alg:fhocp_synthesis} and the dataset $(\mathbf{x},\mathbf{u})$ are reachable to $x_{\agentIdxA,F}$ and satisfy the decomposed constraints $g_{\agentIdxA,t}^\iterIdx(x) \succeq 0$, $\forall x \in \mathcal{SS}_{\agentIdxA,t}^{\iterIdx}$.
\end{proposition}
\begin{proof}
We first assume that $\underline{t}$, $\bar{t}$, and $\mathcal{I}$ are chosen according to Alg. \ref{alg:fhocp_synthesis} such that the constructed sampled safe sets satisfy the decomposed constraints for all sample times $t$. Now, at time $t$ and $t+1$, we obtain the time ranges $\mathcal{T}_t = \{\max(t-\underline{t},0), \dots, t, \dots, t+\bar{t}\}$ and $\mathcal{T}_{t+1} = \{\max(t-\underline{t}+1,0), \dots, t+1, \dots, t+\bar{t}+1\}$, respectively. For the case when $0 \leq t \leq \underline{t}-1$, i.e. both lower limits evaluate to zero, the sampled safe sets at time $t$ and $t+1$ contain $\{x_{\agentIdxA,0}^{\iterIdxB}, \dots, x_{\agentIdxA,t+\bar{t}}^{\iterIdxB}\}$ and $\{x_{\agentIdxA,0}^{\iterIdxB}, \dots, x_{\agentIdxA,t+\bar{t}+1}^{\iterIdxB}\}$, respectively, for all $\iterIdxB \in \mathcal{I}$. Recall that $\mathcal{I}$ contains indices corresponding to successful iterations, i.e. $x_{\agentIdxA,T_{\agentIdxA}^\iterIdxB}^{\iterIdxB} = x_{\agentIdxA,F}, \ \forall \iterIdxB \in \mathcal{I}$. Moreover, by Assumption~\ref{asm:init_feas_traj}, $\mathcal{I}$ must be non-empty. By the fact that the feasible input sequence $\{u_{\agentIdxA,0}^{\iterIdxB}, \dots, u_{\agentIdxA,t+\bar{t}}^{\iterIdxB}\}$ exists in the dataset $\mathbf{u}$, such that $x_{\agentIdxA,k+1}^{\iterIdxB} = f_{\agentIdxA}(x_{\agentIdxA,k}^{\iterIdxB},u_{\agentIdxA,k}^{\iterIdxB}), \ \forall k \in \{0, \dots, t+\bar{t}\}$, we see that the property of reachability to $x_{\agentIdxA,F}$ holds. For the case when $t \geq \underline{t}$, the sampled safe sets at time $t$ and $t+1$ contain $\{x_{\agentIdxA,t-\underline{t}}^{\iterIdxB}, \dots, x_{\agentIdxA,t+\bar{t}}^{\iterIdxB}\}$ and $\{x_{\agentIdxA,t-\underline{t}+1}^{\iterIdxB}, \dots, x_{\agentIdxA,t+\bar{t}+1}^{\iterIdxB}\}$. By the same argument as before, we see that reachability holds. We conclude that the sampled safe sets for agent $\agentIdxA$ at iteration $\iterIdx$ are reachable to $x_{\agentIdxA,F}$ and satisfy the decomposed constraints for all $t \geq 0$.
\end{proof}

\begin{proposition} \label{prop:finite_time_convergence}
Let Assumption~\ref{asm:init_feas_traj} hold and assume that $\mathcal{P}_{\agentIdxA,N}^{\iterIdx}(x_{\agentIdxA,t}^\iterIdx,t)$ defined in \eqref{eq:local_fhocp} is feasible for agent $\agentIdxA$ at time $t$ and iteration $\iterIdx$ for some $x_{\agentIdxA,t}^\iterIdx \in \mathcal{X}_\agentIdxA$. If $x_{\agentIdxA,F} \in \mathcal{SS}_{\agentIdxA,t+N}^{\iterIdx-1}$, then the system \eqref{eq:nonlin_dyn_local} in closed-loop with \eqref{eq:local_fhocp} and \eqref{eq:lmpc_feedback_control_law} converges in at most $t+\bar{T}$ steps to $x_{\agentIdxA,F}$, where $\bar{T} = N+T_\agentIdxA^{\iterIdxB^*}-k^*$ and $(k^*,\iterIdxB^*) = \argmin_{k,p} T_\agentIdxA^\iterIdxB$ subject to the constraint $x_{\agentIdxA,N|t}^{*,\iterIdx}=x_{\agentIdxA,k}^{\iterIdxB}$.
\end{proposition}
\begin{proof}
By assumption, $\mathcal{P}_{\agentIdxA,N}^{\iterIdx}(x_{\agentIdxA,t}^\iterIdx,t)$ is feasible at time $t$, which implies that there exists a sequence of feasible inputs $\bar{\mathbf{u}} = \{\mathbf{u}_{\agentIdxA,t}^{*,\iterIdx}(x_{\agentIdxA,t}^\iterIdx), u_{\agentIdxA,k^*}^{\iterIdxB^*}, \dots, u_{\agentIdxA,T_\agentIdxA^{\iterIdxB^*}-1}^{\iterIdxB^*}\}$, which drives agent $\agentIdxA$ to the goal state $x_{\agentIdxA,F}$ in $\bar{T}$ steps. Therefore,
\begin{align}
    J_{\agentIdxA,N}^{*,\iterIdx}(x_{\agentIdxA,t}^\iterIdx,t) = \bar{T}-1. 
    \label{eq:lmpc_ub}
\end{align}
Since $\mathcal{SS}_{\agentIdxA,t+N}^{\iterIdx-1} \ni x_{\agentIdxA,F}$, it follows from standard MPC arguments and Prop.~\ref{prop:backwards_reachable} that $\mathcal{P}_{\agentIdxA,N}^{\iterIdx}(\cdot,k)$ is feasible for time steps $k > t$ and
\begin{align}
    J_{\agentIdxA,N}^{*,\iterIdx}(x_{\agentIdxA,t}^\iterIdx,t) \geq \sum_{l = t}^k\mathds{1}(x_{\agentIdxA,l}^\iterIdx;x_{\agentIdxA,F}) + J_{\agentIdxA,N}^{*,\iterIdx}(x_{\agentIdxA,k+1}^\iterIdx,k+1),
    \label{eq:lmpc_lb}
\end{align}
where $x_{\agentIdxA,k+1}^\iterIdx=f_\agentIdxA(x_{\agentIdxA,k}^\iterIdx,\bar{\mathbf{u}}_{k-t})$. Assume that $x_{\agentIdxA,k}^\iterIdx \neq x_{\agentIdxA,F}$ for $k = \{t,\dots,t+\bar{T}-1\}$. Then at time step $k = t+\bar{T}-1$, using \eqref{eq:lmpc_ub} and \eqref{eq:lmpc_lb} we have
\begin{align}
    J_{\agentIdxA,N}^{*,\iterIdx}(x_{\agentIdxA,k+1}^\iterIdx,k+1) &\leq J_{\agentIdxA,N}^{*,\iterIdx}(x_{\agentIdxA,t}^\iterIdx,t) - \sum_{l=t}^k \mathds{1}(x_{\agentIdxA,l}^\iterIdx;x_{\agentIdxA,F}) \nonumber \\
    &= (\bar{T}-1) - (\bar{T}-1) = 0 \nonumber
\end{align}
which implies that $x_{\agentIdxA,k+1}^\iterIdx = x_{\agentIdxA,t+\bar{T}}^\iterIdx = x_{\agentIdxA,F}$.
\end{proof}

\begin{theorem}
Consider the system in \eqref{eq:nonlin_dyn_global} in closed loop with the decentralized LMPC \eqref{eq:local_fhocp} and \eqref{eq:lmpc_feedback_control_law}. Let Assumption~\ref{asm:constraint_decomposition} and \ref{asm:init_feas_traj} hold. Then at each iteration $\iterIdx$ for every agent $\agentIdxA \in \agentSet$: 
\begin{enumerate}
    \item The decentralized LMPC \eqref{eq:local_fhocp} is feasible for all time steps $t \geq 0$.
    \item The system \eqref{eq:nonlin_dyn_local} in closed-loop with \eqref{eq:local_fhocp} and \eqref{eq:lmpc_feedback_control_law} converges to the equilibrium point $x_{\agentIdxA,F}$ in finite time.
    \item The sampled safe sets $\mathcal{SS}_{\agentIdxA,t}^{\iterIdx}$ are non-empty for all sample times $t$.
\end{enumerate}
\label{thm:feas_stab}
\end{theorem}

\begin{proof} At iteration $\iterIdx=1$, by Assumption~\ref{asm:constraint_decomposition} and \ref{asm:init_feas_traj}, the sampled safe sets $\mathcal{SS}_{\agentIdxA,t}^{0}$ are non-empty for all $t \in \mathbb{N}$ and (\ref{eq:local_fhocp}) is feasible at $t=0$. 

From Proposition~\ref{prop:backwards_reachable}, we have that at iteration $\iterIdx$, for each agent $\agentIdxA$, the terminal constraint sets $\mathcal{SS}_{\agentIdxA,t}^{\iterIdx-1}$ are reachable to $x_{\agentIdxA,F}$ and satisfy the decomposed constraints for all $t \in \{0,\dots,\max_{\iterIdxB\in\mathcal{I}}(T_\agentIdxA^{\iterIdxB})\}$. Additionally at time $t=0$, \eqref{eq:local_fhocp} is feasible and therefore the proof for 1) follows from standard MPC arguments~\cite{borrelli2017predictive}.

Again, leveraging the reachability property of $\mathcal{SS}_{\agentIdxA,t}^{\iterIdx-1}$, $\exists k \in \{0,\dots,\max_{\iterIdxB\in\mathcal{I}}(T_\agentIdxA^{\iterIdxB})\}$ such that $x_{\agentIdxA,F} \in \mathcal{SS}_{\agentIdxA,k}^{\iterIdx-1}$. Since from 1) we have that \eqref{eq:local_fhocp} is feasible for all $t \geq 0$, 2) follows immediately after applying Proposition~\ref{prop:finite_time_convergence}.

From 1) and 2), task execution at iteration $\iterIdx$ is successful, i.e. $x_{\agentIdxA,T_{\agentIdxA}^{\iterIdx}}^{\iterIdx} = x_{\agentIdxA,F}$. Therefore, by Asssumption~\ref{asm:constraint_decomposition}, we may construct $\mathcal{SS}_{\agentIdxA,t}^{\iterIdx} \ni x_{\agentIdxA,t}^{\iterIdx}$ using Alg.~\ref{alg:fhocp_synthesis}, which shows 3). Repeating the argument for each iteration $\iterIdx$ concludes the proof.
\end{proof}

\begin{corollary}
Consider the system (\ref{eq:nonlin_dyn_global}) in closed loop with the decentralized LMPC (\ref{eq:local_fhocp}) and (\ref{eq:lmpc_feedback_control_law}). Let Assumption \ref{asm:init_feas_traj} hold. Then for every agent $\agentIdxA \in \agentSet$, the task completion time does not increase with the iteration index $\iterIdx$, i.e. $T_{\agentIdxA}^{\iterIdx} \leq T_{\agentIdxA}^{\iterIdx-1}$.
\label{cor:cost_improvement}
\end{corollary}
\begin{proof}
Let $x_{\agentIdxA,t}^{\iterIdx}$ and $u_{\agentIdxA,t}^{\iterIdx}$ be the elements from the closed-loop state trajectory and corresponding input sequence at iteration $\iterIdx$. From (\ref{eq:cost_to_go}), we have that the iteration cost at iteration $\iterIdx-1$ can be written as
\begin{align}
    J_{\agentIdxA,T_{\agentIdxA}^{\iterIdx-1}}^{\iterIdx-1}(x_{\agentIdxA,S},0) &= \sum_{t=0}^{N-1} \mathds{1}(x_{\agentIdxA,t}^{\iterIdx-1};x_{\agentIdxA,F})  + \sum_{t=N}^{T_{\agentIdxA}^{\iterIdx-1}} \mathds{1}(x_{\agentIdxA,t}^{\iterIdx-1};x_{\agentIdxA,F}) \nonumber \\
    &\geq \sum_{t=0}^{N-1} \mathds{1}(x_{\agentIdxA,t}^{\iterIdx-1};x_{\agentIdxA,F}) + V_{\agentIdxA}^{\iterIdx-1}(x_{\agentIdxA,N}^{\iterIdx-1}, N) \nonumber \\
    &\geq J_{\agentIdxA,N}^{*,\iterIdx}(x_{\agentIdxA,S},0).
    \label{eq:cost_improve_ub}
\end{align}
Next, using \eqref{eq:lmpc_lb}, we have that at $t=0$
\begin{align}
    J_{\agentIdxA,N}^{*,\iterIdx}(x_{\agentIdxA,0}^{\iterIdx},0) &\geq \sum_{l=0}^{T_{\agentIdxA}^{\iterIdx}}\mathds{1}(x_{\agentIdxA,l}^{\iterIdx};x_{\agentIdxA,F}) = J_{\agentIdxA,T_{\agentIdxA}^{\iterIdx}}^{\iterIdx}(x_{\agentIdxA,S},0).
    \label{eq:cost_improve_lb}
\end{align}
From (\ref{eq:cost_improve_ub}) and (\ref{eq:cost_improve_lb}), and using the fact that for agent $\agentIdxA$, $x_{\agentIdxA,0}^\iterIdx = x_{\agentIdxA,S}$, we arrive at the final bounds
\begin{align}
    J_{\agentIdxA,T_{\agentIdxA}^{\iterIdx}}^{\iterIdx}(x_{\agentIdxA,S},0)  \leq J_{\agentIdxA,N}^{*,\iterIdx}(x_{\agentIdxA,S},0) \leq J_{\agentIdxA,T_{\agentIdxA}^{\iterIdx-1}}^{\iterIdx-1}(x_{\agentIdxA,S},0), \nonumber
\end{align}
which implies that $T_{\agentIdxA}^{\iterIdx} \leq T_{\agentIdxA}^{\iterIdx-1}$.
\end{proof}

\section{Multi-Vehicle Collision Avoidance} \label{sec:collision_avoidance_example}
In this section, we present a numerical example of decentralized LMPC in the context of multi-vehicle collision avoidance. We compare the results from the decentralized LMPC with those from a centralized approach. The control objective is for $M=3$ vehicles to reach the goal equilibrium points $\zeta_F^\agentIdxA$ from their respective initial states $\zeta_S^\agentIdxA$ in minimum time (the code is available online\footnote{https://github.com/zhu-edward/multi-agent-LMPC}\footnote{A video of the results may be found at https://youtu.be/cB9zckRm5j8}).

\subsection{Agent Model} \label{sec:example_agent_model}

We model the vehicles using the kinematic bicycle model, which is discretized using forward Euler integration with a time step of $dt = 0.1$s as in \cite{kong2015kinematic}.
\begin{align}
    \zeta_{\agentIdxA,t+1} &= \zeta_{\agentIdxA,t} + dt\begin{bmatrix} 
    v_{\agentIdxA,t}\cos(\psi_{\agentIdxA,t} + \beta_{\agentIdxA,t}) \\
    v_{\agentIdxA,t}\sin(\psi_{\agentIdxA,t} + \beta_{\agentIdxA,t}) \\
    v_{\agentIdxA,t}\sin(\beta_{\agentIdxA,t})/l_r \\
    a_{\agentIdxA,t}
    \end{bmatrix} = f_{\agentIdxA}(\zeta_{\agentIdxA,t},u_{\agentIdxA,t}), \nonumber
    \label{eq:agent_model}
\end{align}
where $\beta_{\agentIdxA,t} = \arctan(l_r\tan(\delta_{\agentIdxA,t})/(l_f+l_r))$. 
The state and input variables are $\zeta_{\agentIdxA,t} = [x_{\agentIdxA,t}, y_{\agentIdxA,t}, \psi_{\agentIdxA,t}, v_{\agentIdxA,t}]^{\top} \in \mathbb{R}^4$ (with units m, m, rad, m/s) and $u_{\agentIdxA,t} = [\delta_{\agentIdxA,t}, a_{\agentIdxA,t}]^{\top} \in \mathbb{R}^2$ (with units rad, m/s$^2$) respectively. The vehicles are coupled via collision avoidance constraints where we define a circular collision buffer about the geometric center of the vehicle with radius $r_i$ and require that for all sample times, 
\begin{equation}
    \|\zeta_{\agentIdxA,t}(1{:}2)-\zeta_{\agentIdxB,t}(1{:}2)\|_2 \geq r_i+r_j, \ \forall \agentIdxA,\agentIdxB \in \agentSet, \ \agentIdxA \neq \agentIdxB,
    \label{eq:collision_avoidance_constraint}
\end{equation} 

\subsection{Decoupled FHOCP Formulation} \label{sec:example_fhocp_formulation}

The decoupled FHOCP for agent $\agentIdxA$ at time $t$ and iteration $\iterIdx$ given initial condition $\zeta_{\agentIdxA}$ is formulated as in (\ref{eq:local_fhocp}). 
For the local state and input constraints in (\ref{eq:local_fhocp_local_constr}), in addition to the box constraints $|\zeta_{\agentIdxA,k|t}^{\iterIdx}| \preceq [10, 10, \infty, 10]^{\top}$ and $|u_{\agentIdxA,k|t}^{\iterIdx}| \preceq [0.5, 3]^{\top}, \ \forall k \in \mathbb{N}_{N-1}$, we also impose constraints on the control rate via $|u_{\agentIdxA,0|t}^{\iterIdx} - u_{\agentIdxA,t-1}^{\iterIdx}| \preceq dt\cdot[0.7, 7 ]^{\top}$ and $|u_{\agentIdxA,k+1|t}^{\iterIdx} - u_{\agentIdxA,k|t}^{\iterIdx}| \preceq dt\cdot[0.7, 7]^{\top}, \ \forall k \in \mathbb{N}_{N-2}$. Here, $|\cdot|$ represents the element-wise absolute value.

\begin{figure*}[ht!]
\centering
  \includegraphics[width=0.95\textwidth]{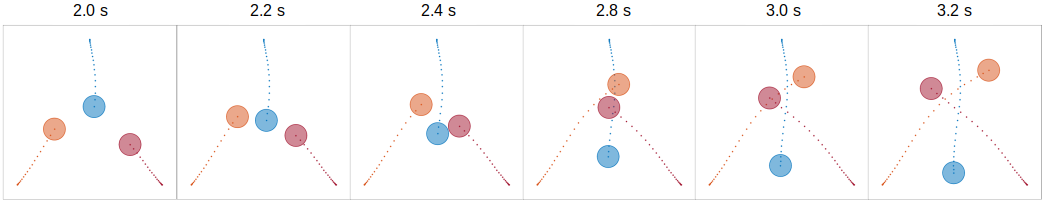}
  \vspace{-0.3cm}
  \caption{Snapshots of the decentralized LMPC trajectory at convergence. Circles represent the collision buffers centered at the position of agent 1 (blue), 2 (orange), and 3 (red).}
  \label{fig:traj_snapshots}
  \vspace{-0.3cm}
\end{figure*}

\begin{figure*}[ht!]
\centering
\includegraphics[width=\textwidth]{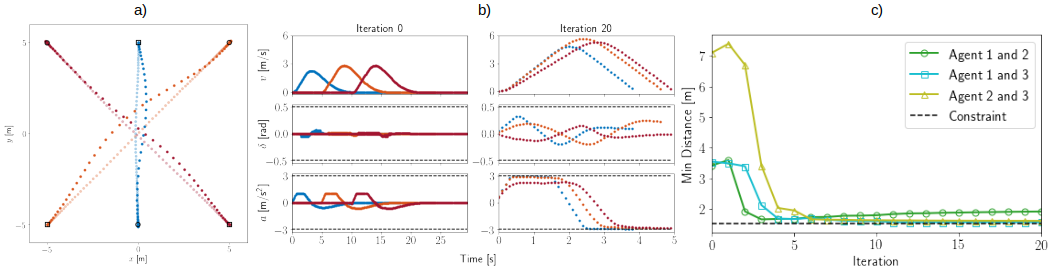}
\vspace{-0.8cm}
\caption{a) Initial (light) vs. steady state (dark) position trajectory (top). The starting and goal states are denoted as the square and circle markers respectively. b) Velocity profile and input sequence (bottom) for agent 1 (blue), 2 (orange), and 3 (red) at the first and last iterations. The black dashed lines correspond to the input box constraints. c) Minimum distance between agents at each iteration.}
\label{fig:combined_results}
\vspace{-0.3cm}
\end{figure*}

In (\ref{eq:local_fhocp_decoup_constr}), we decompose the global constraints into time-varying hyperplane constraints on the position states of each vehicle, i.e. $g_{\agentIdxA,t+k}^{\iterIdx-1}(\zeta_{\agentIdxA,k|t}^{\iterIdx}) = H_{\agentIdxA,t+k}^{\iterIdx-1} \zeta_{\agentIdxA,k|t}^{\iterIdx} + h_{\agentIdxA,t+k}^{\iterIdx-1} \preceq 0$, with $H_{\agentIdxA,t+k}^{\iterIdx-1} \in \mathbb{R}^{M-1\times 4}$ and $h_{\agentIdxA,t+k}^{\iterIdx-1} \in \mathbb{R}^{M-1}$ for all $k \in \mathbb{N}_{N-1}$. 
This implementation of line 8 of Alg.~\ref{alg:fhocp_synthesis} is achieved by solving a hard margin support vector machine (SVM) for all agent pairs $(\agentIdxA,\agentIdxB) \in \agentSet, \ \agentIdxA \neq \agentIdxB$, where maximum margin separating hyperplanes can be found for time $t$ if the position states in the candidate sampled safe sets $\mathcal{SS}_{\agentIdxA,t}^{\iterIdx-1}$ and $\mathcal{SS}_{\agentIdxB,t}^{\iterIdx-1}$ are linearly separable.
If all $\binom{M}{2}$ pairwise SVM problems are feasible for all time $t$ and the distance between the hyperplanes is no less than $2r$ in all cases, then we construct the constraint decomposition for agent $\agentIdxA$ at time $t$ by stacking the solution $H_{ij}$ and $h_i$ from all SVM problems involving agent $\agentIdxA$ into the matrix $H_{\agentIdxA,t}^{\iterIdx-1}$ and vector $h_{\agentIdxA,t}^{\iterIdx-1}$ respectively. Otherwise, following Alg.~\ref{alg:fhocp_synthesis}, we shrink the sets $\mathcal{T}_t$ and $\mathcal{I}$ and retry the SVM problem.

We compute the initial feasible state trajectory and input sequence $({\boldsymbol\zeta}_{\agentIdxA}^{0},\mathbf{u}_{\agentIdxA}^{0})$ using a linear time-varying MPC controller. At each time step, the decoupled FHOCPs are solved using IPOPT \cite{wachter2006implementation} by constructing a set of $|\mathcal{SS}_{\agentIdxA,t+N}^{\iterIdx-1}|$ problems where \eqref{eq:local_fhocp_term_set} is formulated as equality constraints for each element in the safe set. Corollary \ref{cor:cost_improvement} is useful here as it allows us to prune the safe sets using an upper bound on the iteration cost, which can be computed for each point in the safe sets without solving the FHOCP. Specifically, at each time $t$, we have the cost-to-come for the current state $x_{\agentIdxA,t}$, we know that the sum of the stage costs over the optimization horizon is upper bounded by the horizon length $N$, and the cost-to-go from each point in the safe set is known. If for any safe set point, the sum of these three quantities is greater than $T_\agentIdxA^{\iterIdx-1}$, we may discard that point. This helps to manage the computational burden induced by the discrete safe sets. The parameters used for this experiment are shown in Table \ref{tab:params}.

\subsection{Results and Discussion} \label{sec:example_results}




\renewcommand{\arraystretch}{1.1}
\begin{table}[t]
    \centering
    \caption{Parameter Values}
    \begin{tabular}{c|c!{\vrule width 1pt}c|c}
        \toprule
        $\zeta_{1,S}$ & $(0,5,-\pi/2,0)$ & $\zeta_{1,F}$ & $(0,-5,-\pi/2,0)$ \\
        $\zeta_{2,S}$ & $(-5,-5,\pi/4,0)$ & $\zeta_{2,F}$ & $(5,5,\pi/4,0)$ \\
        $\zeta_{3,S}$ & $(5,-5,3\pi/4,0)$ & $\zeta_{3,F}$ & $(-5,5,3\pi/4,0)$ \\
        $l_f,l_r$ & 0.5m & $r$ & 0.75m \\
        $N$ & 20 & $Z$ & 20 \\
        $\epsilon$ & 1e-4 & $\underline{q}$ & 2 \\
        $\bar{t}$ & 175 & $\underline{t}$ & 0 \\ \bottomrule
    \end{tabular}
    \label{tab:params}
\end{table}

\begin{table}[t]
    \centering
    \caption{Optimal Cost of the Decentralized LMPC at each Iteration}
    \begin{tabular}{c c!{\vrule width 1pt}c c}
        \toprule
        Iteration & Iteration Cost & Iteration & Iteration Cost \\ \midrule
        $\iterIdx=0$ & 296 & $\iterIdx=5$ & 51 \\
        $\iterIdx=1$ & 122 & $\iterIdx=6$ & 50 \\
        $\iterIdx=2$ & 79 & $\iterIdx=7$ & 50 \\ 
        $\iterIdx=3$ & 55 & $\iterIdx=8$ & 50 \\
        $\iterIdx=4$ & 51 & &\\
        \bottomrule
    \end{tabular}
    \label{tab:exp_costs}
\end{table}

\begin{table}[t]
    \centering
    \caption{FTOCP Solve Time}
    \begin{tabular}{c|c|c|c}
         & Max Time [s] & Min Time [s] & Avg. Time [s]  \\
        \hline Decentralized & 10.2 & 1.97 & 3.35 \\
        Centralized & 48.3 & 15.8 & 20.5
    \end{tabular}
    \label{tab:time_compare}
\end{table}

As seen in Table~\ref{tab:exp_costs}, the decentralized LMPC converges to a steady state solution where the optimal cost of each iteration is non-increasing. We additionally implemented a centralized LMPC, with the same parameters and solver, for the global system subject to the original constraint (\ref{eq:collision_avoidance_constraint}). This approach achieved a steady state cost of 48, which is only about a 4\% difference in cost with respect to the decentralized case. We also compare the computation time for solving a single FTOCP in the decentralized and centralized cases. This is summarized in Table~\ref{tab:time_compare}. For the former, we record the maximum solve time over agents at each sampling time. For the latter, we record the solve time of the centralized FHOCP at each sampling time. We obtain that over all iterations, computation time of the decentralized case is lower by a factor of 4.6x to 24x.

In Figures \ref{fig:combined_results}a) and \ref{fig:combined_results}b), we compare the initial feasible trajectory to the steady state trajectory at convergence. In the initial feasible trajectory, the agents' movements are intentionally staggered in time to guarantee safety. This can be clearly seen in the velocity profile at iteration 0 in Figure~\ref{fig:combined_results}b). At convergence, all three agents begin moving simultaneously and steer to avoid collisions around the intersection point at the origin. We notice that in the steady state input sequence, the acceleration input either saturates or is close to saturating the imposed constraint and resembles a bang-bang controller \cite{liberzon2011calculus} which switches between acceleration and deceleration at the midpoint of the trajectory. 

In Figure~\ref{fig:traj_snapshots}, we look closely at the steady state trajectory about the intersection point and see that the collision avoidance constraints are satisfied and are almost active for agents 1 and 3 at $2.4$s and agents 2 and 3 at $2.8$s. This is clearly reflected in Figure~\ref{fig:combined_results}c) which plots the minimum pairwise distance between the three agents over iterations of decentralized LMPC.

\section{Conclusion} \label{sec:conclusion}
In this paper, we presented a decentralized LMPC framework for dynamically decoupled multi-agent systems performing iterative tasks. In particular, we proposed a procedure for decomposing global constraints and synthesizing terminal sets and terminal cost functions for the FHOCP using data from previous iterations of task execution. We showed that the resulting decentralized LMPC has the properties of persistent feasibility, finite-time convergence to the goal state, and non-decreasing performance over iterations. 

In the multi-vehicle collision avoidance example, due to the parallelization opportunities afforded by the decentralized implementation, we observe a significant improvement in computation time compared to a centralized approach with only a 4\% increase in cost. In fact, the steady state solution from the decentralized approach saw saturation of the coupling collision avoidance constraint.

Moving forward, we would like to relax the assumption of perfect model knowledge and investigate approaches which leverage techniques in robust and stochastic optimal control to guarantee performance in the presence of model mismatch.

\renewcommand{\baselinestretch}{1.0}
\bibliographystyle{IEEEtran}
\bibliography{main}

\end{document}